\documentclass[10pt,conference]{IEEEtran} 

\usepackage{booktabs} 
\usepackage{balance} 
\usepackage{hyperref}
\usepackage{float}
\floatstyle{boxed}
\newfloat{algorithm}{htbp}{loa}
\floatname{algorithm}{Algorithm}
\newfloat{algorithm}{htbp}{loa}
\usepackage{wrapfig}
\usepackage{multirow}
\usepackage{amsmath,epsfig,epstopdf,amssymb} 
\usepackage{mathtools}

\usepackage{amsmath}
\usepackage{amssymb}

\usepackage{xcolor}
\usepackage{soul}

\usepackage{graphicx}
\usepackage{textcomp}

\usepackage{algpseudocode}
\algnewcommand\algorithmicforeach{\textbf{for each}}
\algdef{S}[FOR]{ForEach}[1]{\algorithmicforeach\ #1\ \algorithmicdo}
\usepackage{tikz-cd}
\usepackage{adjustbox,lipsum}
\usetikzlibrary{arrows}
\usepackage{listings}
 
\newcommand{\rimp}{\Rightarrow}
\title  { Monitoring Cumulative Cost Properties}
 
     \newtheorem{problem}{Problem}
\newtheorem{theorem} {Theorem}

 \newtheorem{definition}{Definition}
 \newtheorem{example}{Example}

 \usepackage{comment}

\author{Omar Al-Bataineh$^{\star}$, Daniel Jun Xian Ng$^{\ast}$, and  Arvind Easwaran$^{\ast}$ \\ $^{\star}$National University of Singapore\\$^{\ast}$Nanyang Technological University}

\date{}

\begin{document}

\thispagestyle{plain}
\pagestyle{plain}

\maketitle

\begin{abstract}

This paper considers the  problem of decentralized monitoring
of a class of non-functional properties (NFPs) with quantitative operators, namely cumulative cost properties. 
The decentralized monitoring of NFPs can be a non-trivial task
for several reasons:
(i) they are typically expressed at a high abstraction level 
where inter-event dependencies are hidden, 
(ii) NFPs are difficult to be monitored in a decentralized way,
and (iii) lack of effective decomposition techniques. We   address   these   issues   by   providing   a formal  framework  for  decentralised  monitoring of LTL formulas with quantitative operators.
The presented  framework  employs the tableau construction and a formula unwinding technique (i.e., a transformation technique that preserves the semantics of the original formula) to split and distribute the input LTL formula and the corresponding quantitative constraint in a way such that monitoring can be performed in a decentralised manner. The employment of these techniques allows  processes to detect early violations of monitored properties and  perform some corrective or recovery actions.  
We demonstrate the effectiveness of the presented
 framework using a case study based on a Fischertechnik training model,
a sorting line which sorts tokens based 
on their color into storage bins. 
The analysis of the case study shows the effectiveness of the  
presented framework not only in early detection of violations, 
but also in developing
 failure recovery plans that can help to avoid
serious impact of failures on the performance of the system.

\end{abstract}

\section{Introduction}

Given the concept of Industry 4.0 \cite{FM4.0}, conventional factories and critical infrastructures evolve into ``smart systems'', which integrate the physical devices and equipment with the cyber communications, creating a critical distributed system. With the growing scale of systems, it is challenging to maintain the stability under all operating conditions. Reducing the downtime and increasing the resiliency to faults become a crucial issue in the system design. Besides, the rapid evolution of systems has led to a significant increase in  systems complexity. This further introduces new challenges in satisfying all the system requirements during the design and execution.

However, to increase the fault tolerance and resilience for distributed systems, many researchers have suggested the use of non-functional properties to evaluate the performance of the systems. An NFP is a specific requirement to evaluate the Quality of Service (QoS) that the system can provide \cite{Chung2009}. For example, execution latency (response time) is a critical NFP since the users normally need to finish a mission in a certain time period. To this end, researchers have designed Assume-Guarantee (A-G) contracts, defined in \cite{Benveniste2012}, to supervise the NPFs of the systems in a centralized fashion.

Unfortunately, for a large-scale distributed systems with numerous processes, one cannot identify the source of the faults whenever the system violates the monitored property (i.e., a formula formalising a requirement over the system's global behaviour which is typically expressed as a Liner Temporal Logic formula). To solve this problem, one can decompose the global formula into simpler sub-formulas. Each sub-formula is monitored by a certain process of the system. Given this decentralized framework, we can rapidly detect the source of a fault if a specific sub-formula fails. However, new challenges also arise in the decentralized framework.

Building a decentralized runtime monitor for a distributed system is a non-trivial task since it involves designing a distributed algorithm that coordinates the monitors in order to reason consistently about the temporal behaviour of the system.  
The formula decomposition techniques can play an important role 
in decentralized monitoring, as it allows the system
to be organized into a set of disjoint groups of processes
where each group is responsible for monitoring
a unique part of the formula.
Formula decomposition techniques can therefore help to improve
the scalability and efficiency of the solution specially when
dealing with large-scale systems.

The main challenge we encounter when developing a decentralized monitoring
solution for distributed systems is how to deal with  properties that are expressed at a high level of abstraction, where inter-event dependencies are hidden.
 An example of such properties is the response time properties of systems, 
which verify the accumulation or difference between the time at which the request occurs
  and the time at which the response is produced.
To address this challenge, we introduce what we call the notion of formula unwinding technique.

The formula unwinding technique
aims at transforming  a system-level formula
into a new formula that is semantically equivalent to the original
formula but makes event dependencies explicit.
The unwinding technique is performed in a way such that
satisfaction/falsification of unwound formula
implies satisfaction/falsification  of the original formula.
The resulting unwound formula is then decomposed into a set of sub-formulas (by tableau decomposition) that reintroduce all intermediate events and modules involved in the monitoring of the initial property. Each sub-formula is then assigned to a process/module for monitoring. 
 The key advantage of the presented monitoring framework is that violations of the monitored formula may be detected far ahead before the actual violation occurs. This allows processes to perform some recovery plans or   corrective  actions
 to  avoid violation of the global formula
 or to mitigate its effect on the entire system.

 \paragraph*{\textbf{Contributions}} We summarize contributions as follows.
 
\begin{itemize}

\item We describe a methodology of creating distributed monitors for monitoring of cumulative cost properties under the assumption
where processes are synchronous and the formula is represented as a tableau.
Specifically, we consider  properties such as execution time, power consumption, memory consumption, etc.
For short we denote such class of properties as CNFPs.

\item  We develop an unwinding algorithm
for CNFPs that can 
be used to transform a system-level formula
into a new formula that is semantically equivalent to the original
formula but makes component event dependencies explicit.
The unwinding algorithm  helps to optimise
 decentralised monitoring
of CNFPs in a way such that violations can  be  detected  way  before  the  original  property would fail.

\item We develop a tableau-based algorithm
that can be used to organise processes of 
a given  system into disjoint groups
where each group can monitor a unique part of the formula.
The developed tableau algorithm helps to reduce  the
 complexity of the  monitoring problem
without compromising soundness. 
The problem of splitting monitoring of systems into simpler monitoring tasks is an interesting research problem, especially when considering applications like cloud, edge and fog computing.

\item  We demonstrate the effectiveness of the presented
  framework for monitoring cumulative cost properties by considering response time properties
 of  systems using a case study  
based on a Fischertechnik training model.
A short  video documentation of the case study is available at \url{https://youtu.be/5CUH0Z2qaBM}.

\end{itemize}

\section{Background}

\subsection{Decentralized Monitoring Problem}

A distributed program $\mathcal{P} = \{p_0,p_1,...,p_{n-1} \}$ is a set of $n$ processes
 working together to achieve a certain task. 
Each process of the system emits events at discrete time instances. Each event $\sigma$ is a set of actions denoted by some atomic propositions from the set $AP$.
We denote $2^{AP}$ by $\Sigma$ and call it the alphabet of the system. We assume that the distributed system operates under the perfect synchrony hypothesis \cite{BauerF12}, and that each process sends and receives messages at \textit{discrete} instances of time, which are represented using identifier $t \in \mathbb{N}^{\geq 0}$.

We assume that each process $p_i$ has a set of input variables 
denoted as $IN(p_i)$ and set of output variables denoted as $OUT(p_i)$.  
We use a projection function $\Pi_i$ to restrict atomic propositions to the local view of monitor $M_i$ attached to process $p_i$, which can only observe events of process $p_i$. For atomic propositions (local to process $p_i$), $\Pi_i: 2^{AP} \rightarrow 2^{AP}$, and we denote $AP_i = \Pi_i (AP)$, for all $i =1...n$. For events, $\Pi_i :2^{\Sigma} \rightarrow 2^{\Sigma}$ and we denote $\Sigma_i = \Pi_i (\Sigma)$ for all $i= 0...n-1$. 
The system's global trace, $g = (g_1, g_2,..., g_n)$ can now be described as a sequence of pair-wise unions of the local events of each process's traces. We denote the set of all possible events in $p_i$ by $E_i$  and the set of all events of $\mathcal{P}$ by $E_{\mathcal{P}} = \bigcup_{i=0}^{n-1} E_i$.
 We assume that the underlying  distributed system is enriched with computational cost: each event $\sigma$ is associated with a cost whose value  depends on the  the running cost of the process that generates that event.
Formally, we assume we have a cost function $\mathcal{C}: E_{\mathcal{P}} \rightarrow  \mathbb{N}$ that maps events of $\mathcal{P}$ to $ \mathbb{N}$, where $ \mathbb{N}$ denotes the set of natural numbers. In our setting,  the assignment of a truth value to a variable is an event and it is the occurrence of this event that we are interested in.
Finally, finite traces over an alphabet $\Sigma$ are denoted by $\Sigma^{*}$, while infinite traces are denoted by $\Sigma^{\infty}$.

\begin{definition}(LTL formulas \cite{Pnueli1977}). 
The set of LTL formulas is inductively defined by the grammar
\[
\varphi ::=  true \mid c \mid \neg \varphi \mid \varphi \lor \varphi \mid X \varphi \mid F \varphi \mid G \varphi \mid \varphi U \varphi  \mid \phi_1 \circ_{\leq q} \phi_2
\]
where $c \in AP$ and $X$ is read as next, $F$ as  eventually (in the future), 
$G$ as  always (globally), and $U$ as until.
Note that we extend the basic LTL language with the metric operator
$\circ_{\leq q}$, namely the quantitative dependency operator.
The operator will be used to express properties with arithmetic constraints.

\end{definition}

\begin{definition} (LTL Semantics \cite{Pnueli1977}). 
Let $ w = a_0 a_1.. . \in \Sigma^{w}$ be an infinite word with $i \in \mathbb{N}$
being a position. Let $d$ be a variable
whose valuation is a mapping from $d$ to $\mathbb{R}^{+}$.
We define the semantics of LTL formulae inductively as follows

\begin{itemize}

\item $w, i \models true$

\item $w, i \models \neg \varphi$ iff $w, i \not\models \varphi$

\item $w, i \models c$ iff $ c \in a_i$
 
\item  $w, i \models \varphi_1 \lor \varphi_2$ iff $w, i \models \varphi_1$ or $w, i \models \varphi_2$

\item $w, i \models F \varphi$ iff $ w, j \models \varphi$ for some $j \geq i$

\item $w, i \models G \varphi$ iff $ w, j \models \varphi$ for all $j \geq i$

\item $w, i \models \varphi_1 U \varphi_2$ iff $\exists_{k \geq i}$ with $w, k \models \varphi_2$ and $\forall_{i \leq l < k}$ with $ w, l \models \varphi_1$

\item  $w, i \models X \varphi$ iff $w, i+1 \models  \varphi$

\item $w, i \models \phi_1 \circ_{\leq q} \phi_2$ iff $ (w, i \models  \phi_1 \land d =x)
\rimp (w, j \models  \phi_2 \land d \leq  (x+q) )$, for some $j \geq i$ and $q \in \mathbb{N}$.

\end{itemize}

\end{definition}

In our setting, a quantitative property is given as an LTL formula extended with a quantitative dependency operator: $\phi_1
\circ_{\leq q} \phi_2$ means that the computation from a state in which $\phi_1$ holds
to a state in which $\phi_2$ holds has a cost bounded by the constraint $q$. The technical challenge of
monitoring quantitative properties in this setting consists of translating global constraints into local ones. This is achieved by computing the maximal cumulative accepted cost for the completion of an
event.
We call the operator $\circ_{\leq q}$ as a quantitative dependency operator and the arithmetic constraint $q$ as CNFP constraint on the cumulative cost.
We call an LTL property that contains the operator $\circ_{\leq q}$ as ``cumulative cost'' property.

\begin{problem} (\textbf{The decentralized monitoring problem}).
Given a distributed system $\mathcal{P} = \{p_0, p_1,..., p_{n-1}\}$,a finite global trace $g \in \Sigma^{*}$, an $LTL$ property $\varphi$ with a set of atomic propositions $AP$
formalising a requirement over the system global behaviour, and a set of monitor processes $ \mathcal{M} = \{ M_1, M_2,..., M_n\}$ such that

\begin{itemize}

\item each process $p_i$ has a local set of propositions $AP_i \subset AP$,

\item each process $p_i$ has a local monitor $M_i$, 

\item each process $p_i$ has a partial view of the global trace $g$, 

\item  monitor $M_i$ can observe local events of $p_i$, 

\item  monitor $M_i$ can communicate with the other monitors.

\end{itemize}

The decentralised monitoring problem aims to design  an  algorithm  for  distributing  and  monitoring
$\varphi$, such that satisfaction or violation of $\varphi$ can be detected by local monitors.
Before proceeding further, let us consider a simple example of a  distributed system and a cumulative cost property by which we demonstrate some of the notions introduced in this section. 
\end{problem}
 

\begin{example}
Suppose we have a distributed system $\mathcal{P}$ with three processes 
$p_0, p_1$ and $p_2$ as described in the  graph given in Fig. \ref{fig:Ex1}. As one can see, there are four variables in the graph $(I_0, O_0, O_1, O_f) $. We call the variable $I_0$ as environment  variable
and the variables $O_0, O_1$ and $O_f$ as dependent variables.
Note that when the truth value of $O_0$ is not issued by process $p_0$ then the truth value of both $O_1$ and $O_f$ will not be issued by processes  $p_1$ and $p_2$ due to the dependency relationships. We assume that the assignment of a truth value to each  variable is associated with a cost which depends on the running costs of the processes $p_0, p_1$ and $p_2$. 
We would like then to monitor the cumulative cost property
$\varphi = G (I_0 \circ_{\leq q} O_f) $. As one can see, the cost is accumulated from one variable to another so that the cost of generating the variable $O_f$ is the sum of individual costs of $O_0$ and $O_1$ and $O_f$. However, to ensure the satisfaction of $\varphi$,  the cumulated costs must not exceed the bound $q$. 
\begin{figure}[h]
  \includegraphics[width=\linewidth]{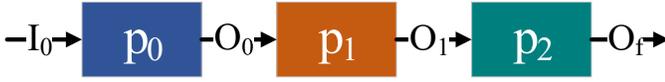}
  \caption{A dependency graph of  a simple system with three processes}
  \label{fig:Ex1}
\end{figure}
\end{example}

\subsection{Tableau Construction for  LTL } \label{sec:tableau}

There are various tableau systems for LTL \cite{Beth59,Smullyan69,Emerson1982,Reynolds2016}.
However, in this work we selected Reynolds's implicit
declarative one \cite{Reynolds2016}. The interesting completeness and termination of the tableau, in addition to its efficiency and simplicity are the key reasons for choosing
this style of tableau.  
Given an LTL formula $\varphi$ we construct a directed graph (tableau) $\mathcal{T}_{\varphi}$ using the standard expansion rules for LTL. 
Applying expansion rules to a formula leads to a new formula but with an equivalent semantics. 
We review here the basic expansion rules of temporal logic:  (1) $G p \equiv p \land XG p$,
(2) $F p \equiv p \lor XF p$, and (3) $p ~ U q \equiv q \lor (p \land X(p~ U q)).$ 
Tableau expansion rules for propositional logic are very straightforward 
and can be described as follows:

\begin{itemize}

\item If a branch of the tableau contains a conjunctive formula $A \land B$, 
add to its leaf the chain of two nodes containing the formulas $A$ and $B$.

\item   If a node on a branch contains a disjunctive formula $A \lor B$, 
then create two sibling children to the leaf of the branch, 
containing $A$ and $B$, respectively.

\end{itemize}

The labels on the tableau proposed by Reynolds are just sets of
formulas from the closure set of the original
formula. 
Note that one can use De Morgan's laws during
the expansion of the tableau, so that for example,
 $\neg (a \land b)$ is treated as $\neg a \lor \neg b$.
A node in $\mathcal{T}_{\varphi}$ is called a leaf if it  has zero children. 
A leaf may be crossed ($\times$), indicating its branch has failed 
(i.e., contains opposite literals), or ticked $\surd$, 
indicating its branch is successful. The whole tableau $\mathcal{T}_{\varphi}$ 
is successful if there is at least a single successful branch.

\begin{figure} [h]
   \begin{minipage}[b]{0.49\linewidth}
    \begin{center}
    \begin{tikzcd}  
 & p \land (q \lor r) \arrow{d}  \\
	 & p, (q \lor r) \arrow{ld} \arrow{rd} \\
    p, q   & & p, r  \\
    \surd  & & \surd
 \end{tikzcd}  
     \end{center}
     \caption{A tableau for  $(p \land (q \lor r))$} \label{Fig:FirstEx}
        \end{minipage}
   \begin{minipage}[b]{0.5\linewidth}
    \begin{center}
    \begin{tikzcd}  
   Gp \arrow{d}   \\
   p, XGp \arrow{d}   \\
   Gp \arrow{d}  \\
      p, XGp   \\
      \surd
 \end{tikzcd}  
     \end{center}
     \caption{A tableau for $Gp$} \label{Fig:GP}
   \end{minipage}
   \end{figure}

Reynolds \cite{Reynolds2016} introduced
a new  tableau rule (the PRUNE rule) which supports a new simple traditional
style tree-shaped tableau for LTL. 
The PRUNE rule provides a simple way to curtail repetitive branch
extension. The PRUNE rule works  as follows. 
If a node at the end of a branch has a label
which has appeared already twice above, and between the second and third appearance there are no new eventualities satisfied that were not already satisfied between the first and second appearances then that whole interval of states (second to third appearance) has been useless. In this case we cut the construction
and declare that the branch is unsuccessful.

Fig. \ref{Fig:FirstEx} represents a tableau for a simple propositional logic formula and Fig. \ref{Fig:GP} represents a tableau for a temporal logic formula.
Using the PRUNE rule and the LOOP rule (a rule that cuts
construction after a poised  label
appears two times in the branch) we 
guarantee completeness and termination of  the tableau construction 
(i.e., it always terminates and returns a semantic graph for the monitored formula
including formulas containing nested temporal operators) \cite{Reynolds2016}.
For example, the formula $G p$ (see Fig. \ref{Fig:GP})
gives rise to a very repetitive infinite tableau without the LOOP rule, 
but succeeds quickly with it.
We first break down the formula into its elementary ones.  
Note that the atoms and their negations can be
satisfied immediately provided there are no contradictions,
 but to reason about the $X$ formula ($ XGp$)
 we need to move forwards in time. Reasoning switches
to the next time point and we carry over only information nested below $X$.

To demonstrate how one can construct a tableau for cumulative cost formulas,
let us construct  the formula $G ((a \land b) \circ_{\leq q} c)$ (see Fig. \ref{Fig:DependencyTableau}).
In the given tableau we use the basic
tableau decomposition rules (the $G$-rule, the $X$-rule, and the $\land$-rule) to decompose the formula in addition 
to the distributive law for the quantitative dependency operator. 
The quantitative operator $\circ_{\leq q}$ satisfies the $\land$-distributive law so that $((a \land b) \circ_{\leq q} c) \equiv (a \circ_{\leq q} c) \land (b \circ_{\leq q} c)$ and the $\lor$-distributive law so that $((a \lor b) \circ_{\leq q} c) \equiv (a \circ_{\leq q} c) \lor (b \circ_{\leq q} c)$.
Note  that we do not decompose  dependency formulas of the form $(a ~ \circ_{\leq q} ~ c)$  as they do not contain temporal or logical connectives. Note also that quantitative dependency formulas of the form $(a ~ \circ_{\leq q} ~ c)$ represent the simplest form of quantitative dependency formulas that maybe encountered when dealing with cumulative cost properties and hence they cannot be split into simpler ones.

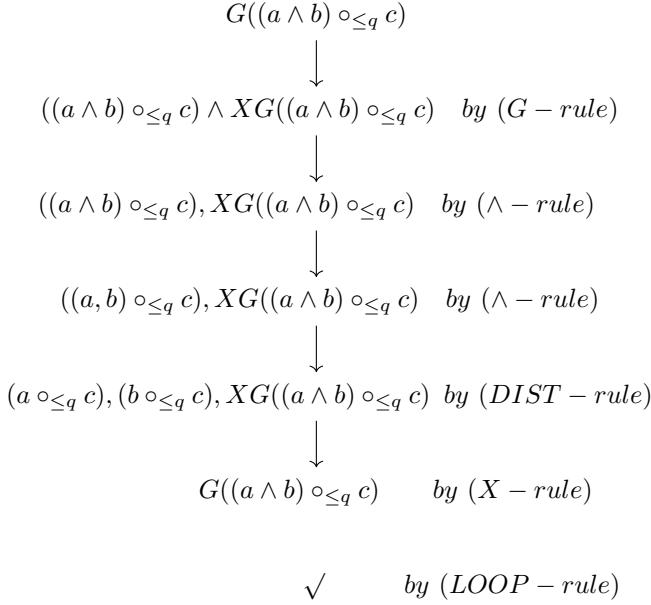
\begin{figure}
  \begin{center}
    \begin{tikzcd}  
  G ((a \land b) \circ_{\leq q} c) \arrow{d}   \\
    \hspace*{10pt} ((a \land b) \circ_{\leq q}  c) \land XG ((a \land b) \circ_{\leq q}  c)\arrow{d}  \hspace*{10pt} by~ (G-rule) \\
    ((a \land b) \circ_{\leq q}  c), XG ((a \land b) \circ_{\leq q}  c) \arrow{d}  \hspace*{10pt}by~ (\land-rule)\\
     \hspace*{10pt}     ((a, b) \circ_{\leq q}  c), XG((a \land b) \circ_{\leq q}  c)  \arrow{d}{}    \hspace*{10pt}by~ (\land-rule) \\
 \hspace*{10pt}    (a \circ_{\leq q}  c), (b \circ_{\leq q}  c), XG((a \land b) \circ_{\leq q}  c)  \arrow{d}{}    \hspace*{5pt}by~ (DIST-rule)\\
    \hspace*{60pt}    G((a \land b) \circ_{\leq q}  c) \hspace*{20pt} by~ (X-rule)  \\ 
       \hspace*{110pt}       \surd \hspace*{30pt} by~ (LOOP-rule) 
 \end{tikzcd}  
     \end{center}
     \caption{A tableau for a cumulative cost formula $G ((a \land b) \circ_{\leq q} c)$} \label{Fig:DependencyTableau}
   \end{figure}


\section{Cumulative Cost Properties} \label{Sec:unfoldingProcess}

NFPs are a class of properties
that are used to express quality attributes of the system.
There is a great variety of NFPs
that can be considered when verifying distributed systems
such as performance, reliability, maintainability and safety.
In this work, we are interested in NFPs that are cumulative in nature
such as response time, energy consumption, memory consumption, etc. 
CNFPs typically contain some constraints 
 related to the running cost of the system. We call such properties as cumulative cost properties (see Definition \ref{CCP}).

\begin{definition} \label{CCP}(\textbf{Cumulative cost properties}).
Let $P $ be a distributed system 
and $\varphi$ be an LTL formula formalising some property of the system $P$. We call the property $\varphi$  a cumulative cost property if  $\varphi$ contains some  quantitative dependency  operator of the form $\circ_{\leq q}$, where $q \in \ \mathbb{N},$ corresponds to the cost cumulated along a running path of $P$ until  certain event is reached denoted by some propositions in $\varphi$.
The manner in which costs are accumulated from one event to another depends on the model  representing the distributed system $P$. 
\end{definition}

Before introducing the notion of unwinding process for cumulative cost properties of systems,
let us discuss first the types of variables that may be encountered when dealing with a distributed system,
which can be classified as follows.
\begin{itemize}

\item \textbf{Independent variables (environment variables)}. 
An independent variable is the variable that is controlled 
and manipulated by the environment. 
 It is independent from the behaviour of the processes of the system.

\item  \textbf{Dependent variables}.  A dependent variable is the variable that is generated
from some process of the system. So that the truth value of the variable
depends on the truth values of some other variables.

\end{itemize}

We assume here we have a dependency graph
of the system that shows the dependency relationships
among its processes. 
We use the dependency graph
to identify dependent variables
and the set of variables that affect their truth values.

\begin{definition} (\textbf{Dependency graph}).
A dependency graph $\mathcal{G}$ of a system  $\mathcal{P}$
is a tuple of the form $(\mathcal{P}, \mathcal{R}, V)$, where

\begin{itemize}
\item $P= \{p_0,..., p_{n-1}\}$ is a set of processes of $P$,

\item  $\mathcal{R} \subseteq \mathcal{P}  \times \mathcal{P} $ is a transition relationship between the processes of the system $\mathcal{P}$,
\item $V = D \cup E$ is the set of variables of the system $P$, where $D$ represents the set of dependent variables and  $E$  represents the set of environment variables.

\end{itemize}

\end{definition}

We assume that a dependency graph does not have any circular dependencies: it forms a directed acyclic graph. Formally,  we require that the transitive closure $\mathcal{R}^{+}$ of the relation $\mathcal{R}$ to be irreflexive; i.e. $(p, p) \not\in \mathcal{R}^{+}$ for all $p \in P$. A pair $(p_i, p_j) \in \mathcal{R}^{+}$ models a dependency (i.e., $p_i$ depends on $p_j$). That is, the output variable issued by $p_i$ depends on the output variable issued by $p_j$.
We classify processes 
in the dependency graph of a given system into three categories as follows.

\begin{enumerate}

\item \textbf{Source processes}. This type of processes have no predecessors
and at least one successor.  The input variables of source processes are called  environment variables.

\item \textbf{Intermediate processes}. This type of processes have at least one predecessor node
and one successor node.

\item \textbf{Sink processes}.  This type of processes have at least one predecessor
and zero successors. The output variables of sink processes represent
the final outputs of the system.

\end{enumerate}

CNFPs are typically given in an abstract form
where inter-variable dependencies are hidden
and hence cannot be efficiently monitored in a decentralised manner. 
To ensure the efficient  monitoring
of CNFPs, the set of intermediate variables need to be explicitly observable
in the formula (being part of the set of propositions of the formula).
To do so,  we compute for each dependent variable what we call 
 the set of dependency paths, which can be extracted from the dependency graph
 of the system. A dependency path for a variable $v$ shows the set of processes and their input and output variables that affect the truth value of the  variable $v$, 
 which is crucial for the unwinding process.
 
Note that the unwinding process of a formula proceeds by unwinding 
dependent variables one-by-one
until all variables are unwound. 
During unwinding we use the following  rules to specify dependency
relationships between variables.
Throughout the rules, we assume that the running cost of the considered process $p$ is bounded by the numerical constraint $q$.

\begin{enumerate}
\item If process $p$ 
takes a single input $I \in IN(p)$ and produces a single output $O \in OUT(p)$
then the resulting dependency formula will take the form $I \circ_{\leq q} O$.
\item If  process $p$ 
takes multiple inputs $(I_1,..., I_k) \in IN(p)$ and produces a single output $O \in OUT(p)$
then   the resulting dependency formula will take the form  
$(I_1 \land... \land I_k) \circ_{\leq q} O$.
\item If  process $p$  
takes a single input $I \in IN(p)$ and produces multiple outputs $(O_1,.., O_k) \in OUT(p)$
then breaking dependencies among variables  will yield $k$ dependency formulae of the form
$(I \circ_{\leq q} O_1, I \circ_{\leq q} O_2,..., I \circ_{\leq q} O_k)$.
\item If  process $p$ 
takes  inputs $(I_1,.., I_k) \in IN(p)$ and produces outputs $(O_1,.., O_m) \in OUT(p)$
then breaking dependencies among variables 
will yield $m$  formulae of the form
$ ((I_1 \land...\land I_k)  \circ_{\leq q} O_1, ((I_1 \land...\land I_k)   \circ_{\leq q}  O_2),.., ((I_1 \land...\land I_k)   \circ _{\leq q} O_m))$.
\end{enumerate}

The decentralised monitoring of LTL formulas can be studied under different  assumptions. However, in this work, we make the following assumptions about
the class of systems and properties that can be monitored by our framework.
\begin{itemize}
\item The system is a synchronous distributed system.
\item The dependency or dataflow graph that highlights all dependencies between modules and input/output variables of the system is available in advance.
\item 
The underlying model (i.e., a distributed system) is augmented with information about cost.
That is, each event in a trace of the system is associated with a numerical value representing the cost of generating that event. 
\item The input formula defines some cumulative cost formula with a quantitative dependency operator of the form $\circ_{\leq q}$.
\end{itemize}

 From the given dependency graph, the initial LTL formula is translated to a set of sub-formulas (by tableau decomposition) that reintroduce all intermediate variables and modules involved in the monitoring of the initial property. Each sub-formula is then assigned to a process/module for monitoring. 
Note that the dependency graph of the system
may contain multiple dependency paths for the 
 dependent variables being unwound and hence the way 
 the running cost of the system is accumulated 
depends heavily on the structure of the dependency graph. 
Recall also that the unwinding process of CNFPs  requires a decomposition 
of the constraints in the formula 
into sub-constraints, which should be performed
while preserving the semantics of the original global formula.
Furthermore, the property of interest may contain
multiple arithmetic constraints related to the  different sub-systems
of the monitored system. 
We address these challenges  at Sections \ref{sec:unwindingAlgor} and \ref{sec:tableaAlgor}.

\section{ Monitoring Framework}

Our monitoring framework for cumulative cost properties consists of two phases: setup and monitor. The setup phase creates the monitors and defines their communication topology. The monitor phase allows the monitors to begin monitoring and propagating information to reach a verdict when possible. We first describe the formal steps of the setup phase.

\begin{itemize}

\item Unwind the original formula $\varphi$ by transforming it into a new formula that is semantically equivalent to  $\varphi$
 but makes variable dependencies explicit.
 We denote the resulting unwound formula by $\varphi^{(U)}$.

\item Negate the unwound formula $\varphi^{(U)}$ 
using the standard LTL negation propagation rules.

\item  Construct a tableau $ \mathcal{T}_{\neg \varphi^{(U)}}$ 
using the method of Sec. \ref{sec:tableau}.

\end{itemize}

The presented decentralised framework consists mainly of two components:
the unwinding component which is described in details at Sec. \ref{sec:unwindingAlgor}
and the decomposition component  which described in details at Sec. \ref{sec:tableaAlgor}.
The unwinding component aims at transforming a system-level formula into a new formula that is semantically equivalent to the original
formula but makes variable dependencies explicit.
 This is crucial for the effectiveness of the
decentralised monitoring of CNFPs.
The decomposition component aims at organising processes into
disjoint groups using tableau.
However, since branches in tableau represent ways 
to satisfy the original formula, we choose
to negate the  formula using LTL negative propagation rules
 before decomposing it using the tableau technique.
In this case, each branch in the constructed tableau 
represents a way to falsify the formula
and therefore violations detected by processes
that monitor a formula representing the semantics of
some  branch in the constructed tableau is a global violation.

Given a distributed system $\mathcal{P} = \{p_0, p_1,..., p_{n-1}\}$, a finite global trace $g = (g_0, g_1,..., g_n) \in \Sigma^{*}$, and an $LTL$ property $\varphi$ 
formalising a requirement over the system $P$
and  $\varphi^{(U)}$ be the unwound version of $\varphi$. We now summarize the  monitoring steps in the form of an algorithm that describes how  process $p_i$ makes decisions regarding the monitored formula  $\varphi^{(U)}$:

\begin{enumerate}
\item $[$Read next event$]$. Read next $\sigma_i \in g_i$ (initially each process reads $\sigma_0$), where $g_i$ is the local trace for  $p_i$.
\item $[$Send new observations$]$. Propagate new observations as pairs of the form $(idx(\phi), val)$ to the successor process, where $idx(\phi)$ 
is the index value of the formula $\phi$ and   $val \in \{true, false, unknown\}$.
\item $[$Receive new observations$]$. Receive new
observations and evaluate the formula $\varphi^{(U)}$.
\item 	$[$Go to step 1$]$. If the trace has not been finished or a decision has not been made then go to step 1. 
\end{enumerate}

To reduce the size of propagated messages, processes send indices of sub-formulas of $\varphi^{(U)}$ that result from the tableau decomposition  rather than formulas themselves. 
That is, we assign a unique index value to each formula in resultant tableau of the unwound formula.
This is possible as variables are pre-known to processes,
thanks to the tableau decomposition.

\subsection{Unwinding Cumulative Cost Properties} \label{sec:unwindingAlgor}

  The unwinding process of CNFPs 
  needs to be performed in a way 
  the semantics of the original formula is preserved.
 Note that the input formula may contain multiple constraints
 with a large number of dependent variables.
 It is necessary then to ensure
 that the unwinding process of a given formula is performed
in a rigorous manner.
We describe here an unwinding algorithm for CNFPs
 which consists of three steps:
 
\begin{enumerate}

\item The preprocessing step.
 The goal of this step is to detect 
dependency operators in the input formula 
and represent each of them as tuples
of the form $(L, R, q)$,
where $L$ is the left operand of $\circ$,
$R$ is the right operand of $\circ$,
and $q$ is the CNFP constraint on the cumulative cost. For example, if the input formula has the form $ G ((c\land d) \circ_{\leq 10} e)$.
Then $L = (c\land d) $, $R = e$, and $q = 10$.

\item  The unwinding step.
The goal of this step is to make all intermediate  
variables that affect the truth value of the original formula
explicitly observable in the unwound formula.
This can be performed by examining the dependency graph of the system under monitoring.

\item The constraint decomposition step. The goal of this step
is to break the arithmetic constraint $q$
 into sub-constraints for different affected sub-formulas in
the unwound formula.

\end{enumerate}

 The unwinding algorithm (Algorithm \ref{alg: unwindingAlgor})
 takes an LTL formula $\varphi$ formalising
a cumulative cost property of interest
together with a dependency graph $\mathcal{G}$ of the system being monitored.
It returns a new formula $\varphi^{(U)}$ 
 in which all intermediate variables become explicitly observable.
Recall that each dependency operator in the formula being analyzed
is represented as a tuple  $(L, R, q)$,
where dependent variables  are unwound first
and then the constraint $q$ is decomposed while taking into consideration
the dependency relationships among variables and the running costs of
 processes.

During the unwinding process,
the algorithm replaces each dependent variable
by its full dependency formula (the set of variables that affect its truth value)
as derived from the dependency graph of the system being monitored.
Such replacement is performed while preserving the semantics of the original formula.
The function $DependencyPath (v_i)$ is a function
that returns the set of processes along the dependency paths of the variable $v_i$.
The function $Cost (Path)$
returns the sum of the running costs of the processes 
along the path $path$. 
Intuitively, for a path of $n$ processes $p_0,..., p_{n-1}$, 
we have  $Cost (path) = \sum_{i =0} ^{n-1} (cost(p_i))$.
Hence, the constraint associated with the dependency 
formula $\phi_i$ assigned to process $p_i$ is synthesized
using the following formula
\begin{equation} \label{SynthFormula}
c_i = q - (\sum_{j = i +1} ^{n-1} (cost (p_j)))
\end{equation}
where $q$ is an arithmetic constraint given in the original formula.
Note that Formula (\ref{SynthFormula})
takes advantage of the fact that the property being monitored
has an additive nature and hence the running cost of the system
accumulates along the  paths.
We can therefore decompose the constraint $q$
into sub-constraints by considering the running costs
of local processes.
Note that it is possible to have more than
one dependency path that leads from process $p_i$ to the process
that produces the variable being unwound.
In this case, the parameter $c_i$ is computed
by considering the path with the least cost.

\begin{algorithm}  [h!]
\begin{algorithmic}[1]
\State \textbf{Inputs} : $(\varphi, \mathcal{G})$
\State \textbf{Output} : $\varphi^{(U)} := \varphi$ 
\State $list := \emptyset$
\State $Queue$ $Waiting := \emptyset$
\ForEach {$\circ \in Operators(\varphi)$} \Comment {Preprocessing phase}
\State $L = getLeftOPND (\circ )$
\State $R = getRightOPND (\circ )$
\State $q = getConstraint (\circ )$
\State \textbf{add} ($L, R,  q$) \textbf{to} $Waiting$
\EndFor
\While{$Waiting \neq \emptyset$}  \Comment{Unwinding phase}
\State \textbf{select} $(L, R, q)$ \textbf{from} $Waiting$
\State $finalFormula := true$
\State $unwind:= false$
\ForEach{$v_i \in atoms (R)$}
\State \textbf{add} $v_i$ \textbf{to} $list$
\State $\psi := true$
\While {$ list \neq \emptyset$}
\State \textbf{select} $v_i$ from $list$
\If {$v_i \in OUT(p) \mid p \in processes(\mathcal{G})$} 
\State $unwind:= true$
\State $paths := DependencyPath (v_i)$
 \State $Val := \min (\forall_{path \in paths}  (Cost (path))$
\State $c_i := (q - Val)$  
\State $\psi := (\bigwedge_{i =1...n} (I_i \mid I_i \in IN (p))) \circ_{\leq c_i} v_i)$
\State $\psi^{'} := \psi^{'} \land \psi$
\EndIf
\ForEach{$v_j \in (atoms(\psi) \setminus v_i$)}
\If {$v_j \in OUT(p) \mid p \in processes(\mathcal{G})$}
\State \textbf{add} $v_j$ \textbf{to} $list$
\EndIf
\EndFor 
\EndWhile
\State $finalFormula := (finalFormula \land \psi^{'})$
\EndFor
\If{$unwind = true$}
\State \textbf{replace}  $(L \circ_{\leq q} R)$ by $finalFormul$ in $\varphi^{(U)}$  
\EndIf
\EndWhile
\State \textbf{return} $\varphi^{(U)}$
\end{algorithmic}
\caption{Unwinding   cumulative cost formulas}  \label{alg: unwindingAlgor}
\label{alg:computingMinSet}
\end{algorithm}

\subsection{Organizing Processes into Disjoint Groups} \label{sec:tableaAlgor}

Approaches to decomposition of formulas can be classified into logical approaches
and algebraic approaches. The first are based on equivalent transformations
of formulas in propositional or temporal logic.
The second ones consider formulas as algebraic objects	
with corresponding transformation rules.
In this work, we follow the logical approach of formula decomposition 
and we adopt the tableau technique for this purpose.
It is advantageous to use tableau as a decomposition technique for decentralised monitoring. 
First, it can be used to detect tautological and unsatisfiable parts
of the formula and to propagate information about only feasible branches.
Second, it helps to reduce the complexity of the  monitoring problem.

\begin{definition} \label{ANDDecomp}(\textbf{Decomposability}). 
An LTL formula 	$\varphi$ is called disjointly OR-decomposable (or decomposable, for short)
wrt a system $P$ if it is equivalent to the disjunction $\phi_1 \lor \phi_2 \lor ... \lor \phi_n$
of some formulas $\phi_1,.., \phi_n$ such that:

\begin{enumerate}

\item $atoms(\phi_1)  \cup... \cup ~ atoms(\phi_n) = atoms (\varphi)$, where $ n > 1$;

\item $atoms(\phi_i) \neq \emptyset$, for $i =1...n$;

\item $AP_p \cap atoms(\phi_i) \cap atoms(\phi_j) = \emptyset$, for any $p \in P$ and
$i\neq j, i, j = 1..,n$.

\end{enumerate}
\end{definition}

where $atoms(\phi_i)$ represents the set of atomic propositions in $\phi_i$
and $AP_p$  represents the set of atomic propositions 
that are locally observed by process $p$. 
The formulas $\phi_1,.., \phi_n$ are called decomposition components of $\varphi$.
The variable sets of the components must be proper subsets of the variables of
the original formula $\varphi$. The obtained formulas define 
some partition of $atoms(\varphi)$ that is observed by 
a unique subset of processes in order to ensure disjointness.

In this work, we view a tableau $\mathcal{T}_{\varphi}$ of an LTL formula $\varphi$ 
as a set of branches $\mathcal{B}_1,..., \mathcal{B}_k$ where
each branch $\mathcal{B}_i$ consists of a sequence of nodes
$(n_0,..., n_{\ell})$, where  $n_0 = \varphi$ is the initial node
and $n_{\ell}$ is the leaf or terminal node of the branch $\mathcal{B}_i$.
The formulas at node $n_{\ell}$ are generated through the repeated application of the tableau decomposition rules
and hence they are either in their simplest form (atomic formulas)
or that no new information can be obtained from  decomposing further the formulas
(a fixed point has been reached).
Hence, we need only to examine terminal nodes of branches
when organizing processes into groups using the tableau representation.

We now  describe a formula decomposition  
algorithm  (Algorithm \ref{alg:SplittingGroups})
that can be used to perform a logical decomposition 
of the formula based on the observation power of processes
and the tableau representation of the formula.
Note that the unwinding algorithm performs a decomposition of
the constraints in the formula but not a logical decomposition of the formula itself, 
which will be performed by the tableau algorithm presented here.
 The tableau algorithm takes as inputs the parameters $(\mathcal{P}, \mathcal{T}_{\varphi})$,
 and returns a set of groups of processes
 with their corresponding assigned LTL formulas $(group_{1}, \phi_1),..., (group_{n}, \phi_{n})$. 
The function 
$GetTerminalNode(\mathcal{B})$ returns
the terminal node (set of formulas at the last node) in the branch $\mathcal{B}$.
The algorithm consists of two phases: the exploring phase and the merging phase.
In the exploring phase, the branches of the tableau 
are examined in order to compute the set of processes
that contribute to their truth values.
In the merging phase, joint groups (groups with common processes) are merged.
This is necessary in order to avoid communications across groups.
We assume that processes within groups communicate with each other
using a static communication scheme in which the order
of communication is determined by their PIDs. 
\begin{algorithm} [h!]
\begin{algorithmic}[1]
 \caption{Organizing processes into disjoint groups} 
\State \textbf{Input}: $(\mathcal{P}, \mathcal{T}_{\varphi})$ 
\State $group := \emptyset, ListOfGroups := \emptyset $
\State \textbf{Output}: $ListOfGroups $
\If{$|\mathcal{T}_{\varphi}|  = 1$}  \textbf{return} $\{(P, \varphi)\}$
\EndIf
\ForEach{$\mathcal{B} \in \mathcal{T}_{\varphi}$ } \Comment {Exploring phase}
\State $NT = GetTerminalNode (\mathcal{B})$
\State $\phi := \bigwedge_{ i = 0} ^{|NT|} (NT_i)$
\ForEach {$p \in P$}
\If{$(AP_p \cap atoms(\phi))  \neq \emptyset$}  
\State \textbf{add} $p$ \textbf{to} $group$
\EndIf
\EndFor
\State \textbf{add} $(group, \phi)$ \textbf{to} $ListOfGroups $
\State $group := \emptyset$
\EndFor 
\ForEach{$M \in  ListOfGroups$}   \Comment{Merging phase}
\ForEach{$N \in  ListOfGroups \setminus M$}
\If{$M.group \cap N.group \neq \emptyset$}
\State $M.group := merge(M.group, N.group) $
\State $M.\phi := M.\phi \lor N.\phi$
\State \textbf{remove} $N$ \textbf{from} $ListOfGroups $
\State \textbf{add} $M$ \textbf{to} $ListOfGroups $
\EndIf
\EndFor
\EndFor
\State \textbf{return} $ListOfGroups $
\label{alg:SplittingGroups}
\end{algorithmic}
\end{algorithm}

To show how one can  monitor 
CNFPs in a decentralised manner,
we consider response time properties 
as an example.

\begin{example}

Suppose we have a system
that consists of 7 processes $(p_0, p_1, p_2, p_3, p_4, p_5, p_6)$
as shown in Fig. \ref{fig:Ex2}. 
\begin{figure}[h]
  \includegraphics[width=\linewidth]{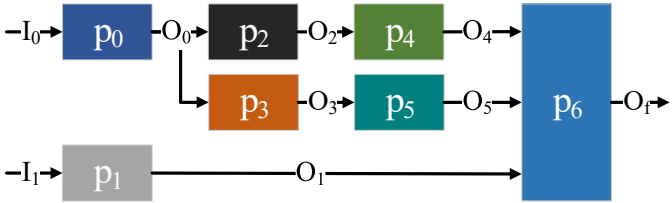}
  \caption{A system with multiple dependency paths.}
  \label{fig:Ex2}
\end{figure}
The property that we would like
to monitor in a decentralized manner for the given system
is $\varphi = G ((I_0 \land I_1)  \circ_{\leq 20} O_f)$.
Obviously, the formula in its given form cannot be monitored efficiently
in a decentralized way since it is given in
an abstract form where all inter-dependent variables are hidden.
We therefore need first to unwind the formula $\varphi$.
This can be performed by examining the dependency graph of the system.
The resulting unwound formula has the following form
\begin{equation} \label{equation}
\begin{array}[t]{l}
\varphi^{(U)} = G ( (O_1 \land O_4 \land O_5)   \circ_{\leq c_6} O_f) \land 
G (O_2  \circ_{\leq c_5} O_4) \land \\
\hspace*{35pt}G (O_3  \circ_{\leq c_4} O_5) \land G (O_0  \circ_{\leq c_3} O_2) \land G (O_0  \circ_{\leq c_2} O_3) \land \\
\hspace*{35pt} G (I_1  \circ_{\leq c_1} O_1) \land G(I_0  \circ_{\leq c_0} O_0).
\end{array}
\end{equation}
Formula (\ref{equation}) can be negated as follows
\begin{equation} \label{equation1}
\resizebox{\linewidth}{!} {
\begin{array}[t]{l}
\neg \varphi^{(U)} = F \neg ((O_1 \land O_4 \land O_5)   \circ_{\leq c_6} O_f) \lor 
F \neg (O_2  \circ_{\leq c_5} O_4) \lor \\
\hspace*{30pt} F \neg(O_3  \circ_{\leq c_4} O_5) \lor F \neg (O_0  \circ_{\leq c_3}{} O_2) \lor F \neg (O_0  \circ_{\leq c_2} O_3) \\ 
\hspace*{30pt} \lor  F \neg (I_1 \circ_{\leq c_1} O_1) \lor F \neg (I_0 \circ_{\leq c_0} O_0).
\end{array}}
\end{equation}

\begin{table} 
\caption{Sub-formulas and their corresponding assigned processes as generated by Algorithm \ref{alg:SplittingGroups}} \label{table: tableResult1}
\begin{center}
\begin{tabular}{ |c|c|}
\hline
Sub-formula & Monitoring process \\
\hline
$ F \neg (I_0 \circ_{\leq c_0} O_0)$ & $p_0$ \\
\hline
$ F \neg (I_1 \circ_{\leq c_1} O_1)$ & $p_1$ \\
\hline
$ F \neg (O_0 \circ_{\leq c_2} O_3)$ & $p_2$ \\
\hline
$ F \neg (O_0 \circ_{\leq c_3} O_2)$ & $p_3$ \\
\hline
$ F \neg (O_2 \circ_{\leq c_4} O_4)$ & $p_4$ \\
\hline
$ F \neg(O_3 \circ_{\leq c_5}  O_5 )$ & $p_5$ \\
\hline
$ F \neg ((O_1 \land O_4 \land O_5) \circ_{\leq c_6} O_f)$ & $p_6$ \\
\hline
\end{tabular}
\end{center}
\end{table}

Formula (\ref{equation1}) is then decomposed using the 
 tableau technique.
The resulting tableau of this formula consists of six branches
where each branch represents a way to falsify the original formula. 
We then use  Algorithm \ref{alg:SplittingGroups} 
to organise processes into disjoint groups
as described in Table \ref{table: tableResult1}.
Note that for this particular example
processes need not to communicate with each other
and they can detect violation of the monitored formula (if any) separately.
This is mainly due to the syntactic structure of the given formula.
Thanks to the tableau decomposition!

The attribution of processes to each sub-formulas for monitoring relies mainly on the observation power of processes (the set of variables that are locally observed by each process). For example, process $p_0$ is the only process among processes that can observe (locally) the variables $I_0$ and $O_0$ and hence the first formula in Table   \ref{table: tableResult1} is assigned to $p_0$.
The constraints $c_0,...,c_6$ can be computed using formula (\ref{SynthFormula}) as follows
 $$
\begin{array}[t]{l}
c_0 = \min ( 20 - cost (p_2) + cost (p_4) + cost (p_6)), \\
 \hspace*{40pt}		   (20- cost (p_3) + cost (p_5)+cost (p_6)))
\end{array}
$$
$$
c_1 = (20 - cost (p_6) ); ~
c_2 = (20 -  (cost (p_5) + cost (p_6)))
$$
$$
c_3 = (20 -  (cost (p_4) + cost (p_6))); ~
c_4 = (20 - cost(p_6))
$$
$$
c_5 = (20 - cost (p_6) ); ~
c_6 = 20.
$$

\end{example}

Suppose that the lower running costs (response times)
of processes are given as follows:
$cost (p_0) = 2, cost (p_1) = 3, cost (p_2) =1, cost (p_3) =2, cost (p_4) = 4,
 cost (p_5) = 3, cost (p_6) = 4$.
The values of the sub-constraints assigned to the processes will be as follows
$$
c_0 = 11 ;~ c_1 = 16; ~c_2 = 13;  c_3 = 12;
c_4 = c_5 = 16; ~c_6=  20.
$$
Hence, the earliest possible time at which
violation (if any) of the property $\varphi$
can be detected will be at $(x + 11)$,
where $x$ represents the time at which monitoring
has been initiated.

\subsection{The Soundness of Monitoring Framework}

By assuming that the dependency graph of the system
is finite, one can show that the formula
$\varphi$ can be unwound  in a finite number of unwinding steps.
An upper bound on the number of unwinding steps 
can be computed in terms of the number of processes
and the number of dependent variables in $\varphi$. 
Termination of Algorithm \ref{alg: unwindingAlgor} is guaranteed since we assume that the dependency graph $\mathcal{G}$ does not have any circular dependencies. 
In Theorem 1, we show that the transformation (unwinding) of the input formula into a new formula that makes variable dependencies explicit is sound. That is, the original formula and the unwound formula are semantically equivalent. The soundness of transformation relies heavily on the employed graph traversal strategy that is used to unwind dependent variables in the input LTL formula.
The traversal strategy needs to respect the order at which intermediate variables are generated. 
This implies that monitoring of the original formula and the unwound formula  yields the same verdict. One of the key advantages of monitoring the unwound (extended) formula over the original (abstract) formula is that violations maybe detected way  before  the  original  property would fail and hence some corrective actions maybe taken to avoid severe consequences of failure.

\begin{theorem} (\textbf{Soundness of unwinding})
Let $\mathcal{G}$ be a dependency graph for a  system $\mathcal{P}$ and $\varphi$ be an LTL property formalising a cumulative cost property of $\mathcal{P}$.
Let $\varphi^{(U)}$ be an LTL formula obtained
by unwinding the property $\varphi$ using Algorithm \ref{alg: unwindingAlgor}.
Then $\varphi$ and $\varphi^{(U)}$ are semantically equivalent.
\end{theorem}

\begin{proof}.
Let $\varphi$ be an LTL formula formalising a cumulative cost property with a quantitative dependency constraint of the form $(L \circ_{\leq q} R)$, where
 $q \in \mathbb{N}$.
Let also  $\mathcal{G}$ be the dependency graph of the system $\mathcal{P}$ and $\varphi^{(U)}$ be an unwound version of $\varphi$
with the set of constraints $\{c_1,..., c_n\}$  obtained by running Algorithm \ref{alg: unwindingAlgor}.
Suppose that $Var^{\varphi}$ and  $Var^{\varphi^{(U)}}$ 
are the set of variables in $\varphi$ and $\varphi^{(U)}$ respectively.
To prove the theorem we need to show 
that the construction of $\varphi^{(U)}$
from  $\varphi$ and $\mathcal{G}$ (Algorithm \ref{alg: unwindingAlgor}) meets the following correctness criteria: (1)  the unwinding of dependent variables in $\varphi$ using the graph $\mathcal{G}$ preserves the semantics of the property $\varphi$, and (2) the decomposition of the global constraint $q$ into local constraints respects the order at which intermediate variables are generated.
To show that Algorithm 1 meets the first criterion 
let us consider a dependent variable $v$ in $(L \circ_{\leq q} R)$.
To unwind the variable $v$, Algorithm \ref{alg: unwindingAlgor} conducts a backward analysis of the graph $\mathcal{G}$ starting from the process that generates $v$ until it reaches some source process (i.e., a process whose inputs are independent or environment inputs) (see lines 18-35). 
Note that some of intermediate variables along the explored dependency path affect the truth value of the variable $v$
(i.e., if truth values of intermediate variables are missing then the truth value of $v$ cannot be obtained).
Algorithm \ref{alg: unwindingAlgor} then constructs a full dependency formula for the explored path(s) which takes the form 
$(\phi_1 \circ_{\leq c_1} \phi_2) \land ... \land (\phi_{n-1} \circ_{\leq c_n} \phi_{n}$),
where $\phi_i$ can be either atomic formula or compound formula and $n$ represents the number of processes along the visited dependency path.  The above steps are repeated on each detected dependent variables in  $(L \circ_{\leq q} R)$. 
Finally, Algorithm \ref{alg: unwindingAlgor} replaces the quantitative dependency formula $(L \circ_{\leq q} R)$
 under analysis with the resultant unwound quantitative  dependency formula to conclude the unwinding process (see lines 36-37). 
It is easy to see that traversing the graph $\mathcal{G}$ in this manner (backward traversing)
that respects the order at which intermediate variables are generated ensures soundness of transformation.
The constraint $q$ is decomposed among monitoring processes using formula (\ref{SynthFormula}).
The constraint $q$ is decomposed into  local constraints in a way such that violations detected by individual processes are actual violations. To do so, we need to ensure that the constraint $c_i$ assigned to process $p_i$  represents the maximal cumulative accepted cost for the completion of an event generated by that process. To achieve this, Algorithm 1 subtracts the global constraint $q$ from the cost of the path that have the minimal cumulated cost  among all paths that lead from  process $p_i$ to  the  process  that generates the variable being unwound (see lines 22-24).
It is easy to see that such decomposition of the constraint $q$ into sub-constraints $c_1,...,c_n$ preserves the semantics of the original quantitative formula.
 Hence, under the same truth assignments of variables in $(Var^{\varphi} \cap Var^{\varphi^{(U)}})$,
formulas $\varphi$ with the constraint $q$ and  $\varphi^{(U)}$ with the constraints $\{c_1,..., c_n\} $ yield the same output. 
\end{proof}

\begin{theorem} \label{soundness} (\textbf{Soundness  of monitoring}). 
Let $\varphi \in LTL$  formalising a cumulative cost property  of a system $\mathcal{P}$ and $g \in \Sigma^{*}$ be a global trace.
Let $\varphi^{(U)}$ be the unwound version of $\varphi$. 
Then $g \models \varphi^{(U)} :B \rimp g \models \varphi : B$, where $B \in \{\top, \bot\}$.
\end{theorem}

\begin{proof}
 Theorem 2 is a direct implication of theorem 1 as formulas $\varphi$ and $\varphi^{(U)}$ are semantically equivalent.
\end{proof}

\section{A Case Study}

\subsection{A description of The Case Study}

The case study presented here is based on a Fischertechnik training model which we use to demonstrate the  advantages of the unwinding approach and the underlying monitoring framework. This model factory, as shown in Fig. ~\ref{fig:testbed}, is a sorting line which sorts tokens based 
on their color into storage bins. 

\begin{figure}[htbp]
\small
	\centering
	\includegraphics[page=1, width=\columnwidth]{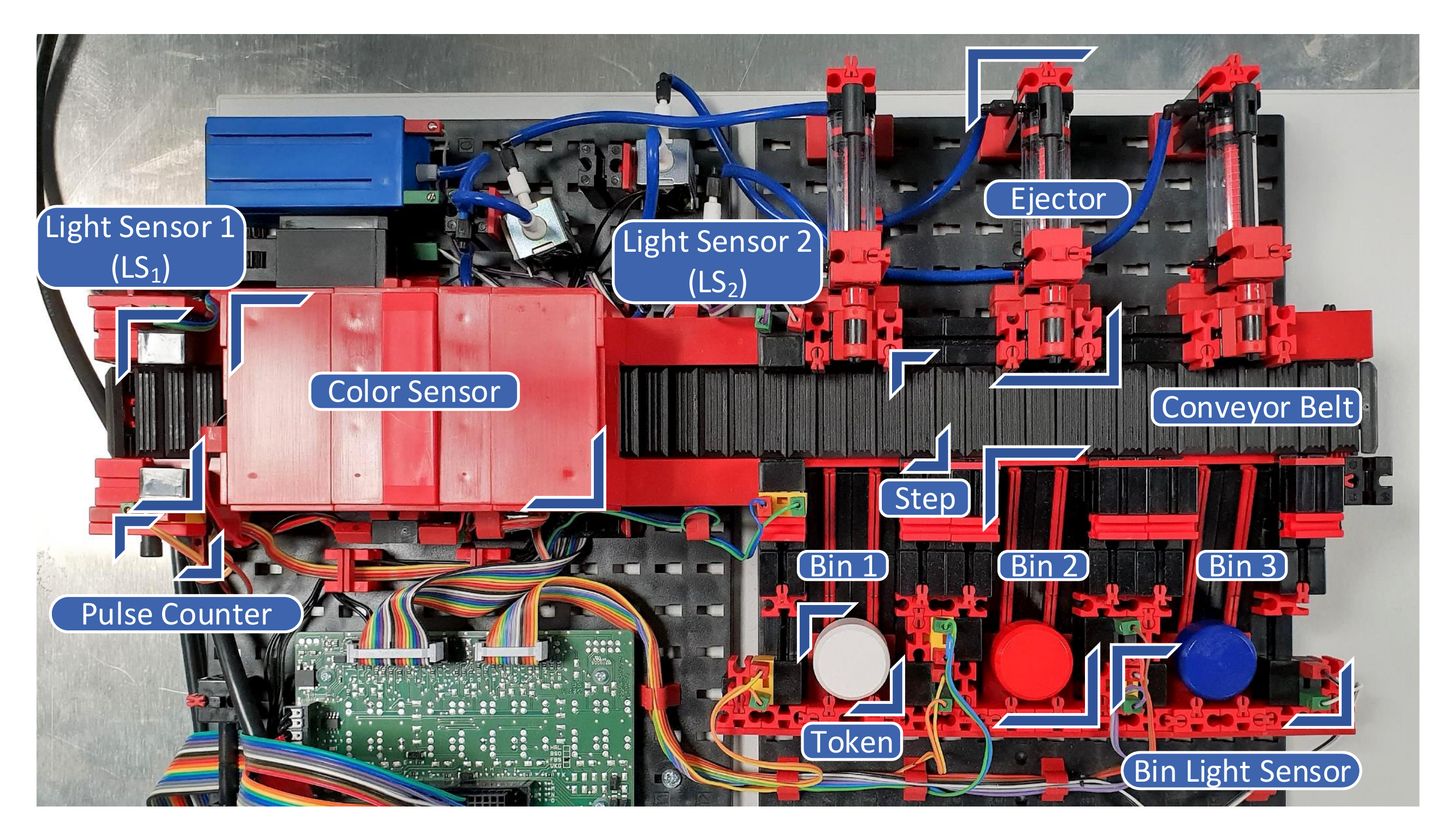}
	\caption{Fischertechnik Testbed: Sorting line with color detection.}
	\label{fig:testbed}
\end{figure}

The processes of the model factory, including its actuators and sensors,
are given below.
\begin{itemize}
	\item \textbf{Light sensors: } Two light sensors for the detection of a token on the conveyor belt.
	\item \textbf{Color sensor: } This sensor provides an analog signal for color determination of a token. 
	\item \textbf{Ejector: } One of three ejectors is used to push the color sorted token into the storage bins.
	\item \textbf{Storage bins: } There are three storage bins where each has a sensor.
	\item \textbf{Direct current (DC) motor: } This motor is responsible
	for providing the power necessary for the rotation of the belt.
	\item \textbf{Pulse counter: } An encoder to track the movement of the conveyor belt through step counts.
	\item \textbf{Conveyor belt: } This is a physical belt which moves the token to its bin.
	\item \textbf{Tokens: } There are two types of token one is a white token and the other is a blue token.
\end{itemize}

\begin{figure}[htbp]
\small
	\centering
	\includegraphics[page=2, width=3.5in]{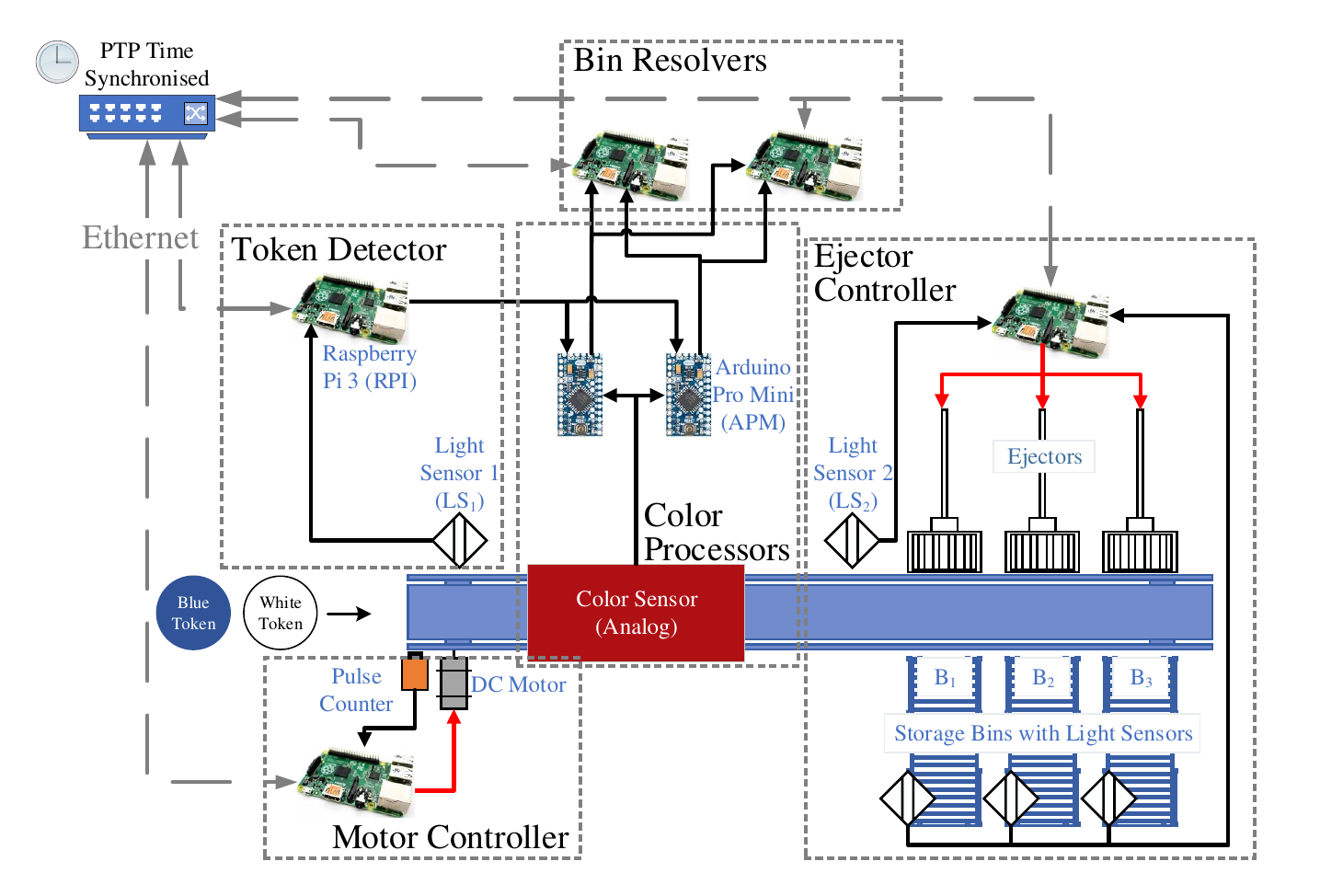}
	\caption{Hardware Architecture: The Fischertechnik training model with the various processes grouped together: A Token Detector, Color Processors, bin Resolvers and the Ejector Controller.}
	\label{fig:hw_arch}
\end{figure}

The various processes of the model factory are shown in Figure~\ref{fig:hw_arch}. A token first enters the conveyor belt from the left side and is then detected by the first light sensor $ LS_1 $. It moves along the conveyor belt and reaches the color sensor which then identifies the color of the token (i.e., white or blue). As it moves along the conveyor belt, the token passes through the second light sensor $ LS_2 $. Then after passing the light sensor $ LS_2 $, the ejectors then eject the token into one of the three bins ($ B_1$, $B_2 $ or $B_3 $). The   bin $ B_1 $ is designated for the white token while the bin $ B_2 $ is designated
for the blue token. The movement of a token is tracked through the pulse counter
which counts the number of steps the token made on the  conveyor belt.

To control the sorting line, a collective of Raspberry Pi (RPI) 3s are used as computation nodes while the Arduino Pro Minis (APMs) are used as analogue to digital converters to process the analogue signal from the color sensor. Once the analogue signal is processed, the color information is communicated to the RPI. A motor controller (MC) regulates the DC motor which in turn regulates the belt's rotation and also tracks the belt's steps through the pulse counter. The token detector (TD) monitors the arrival of tokens through $ LS_1 $ and triggers the color processors (WCP and BCP) to read the color sensor. Both of them are used to process the analogue value from the color sensor to determine the color of the token.  There are two managerial processes, the bin resolvers (WBR and BBR) which receive the color output from their CPs and determine the token's bin placement. As the token's color is being read while it is moving, the color sensor produces a noisy analogue value, leading to inaccurate color readings. Each color processor in the system is designed to be biased towards their assigned color to combat this problem. After that, the ejector controller (EC) receives bin information from the BRs and triggers the corresponding ejector. EC also monitors inputs from $ LS_2 $ to ensure the timely arrival of the token. The dependencies among processes can be seen in Figure~\ref{fig:dep_graph}.

The application programs on the RPIs are written with 4DIAC, which is based on the IEC 61499 standard~\cite{4DIAC}. RPIs in the system are networked through Ethernet and are also time synchronized through Precision Time Protocol to enable decentralized monitoring through the unwinding technique.

\begin{figure}[htbp]
\small
	\centering
	\includegraphics[page=3, width= 3.5in]{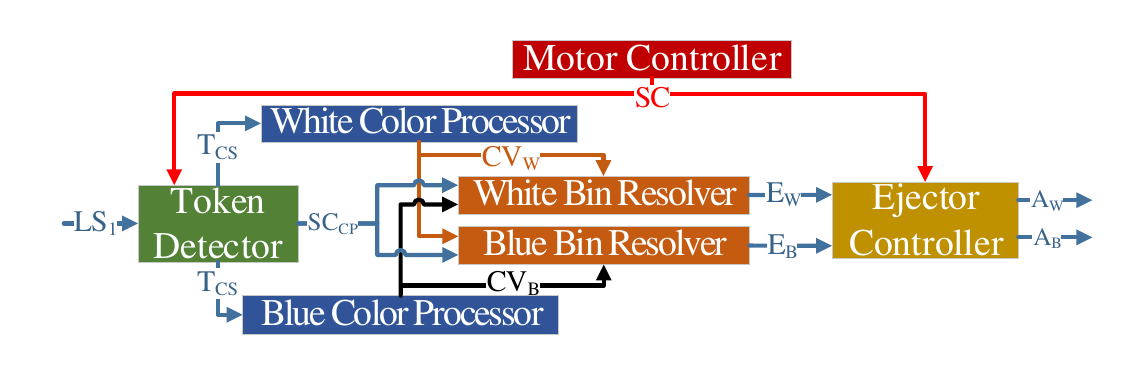}
	\caption{Dependency Graph \ensuremath{\mathcal{G}}: Shows the input-output dependencies of the various processes in the system.}
	\label{fig:dep_graph}
\end{figure}

The dependency graph $ \ensuremath{\mathcal{G}} $ illustrated in Figure~\ref{fig:dep_graph} shows an end-to-end timing requirement for all the processes such that both the white and the blue tokens can be sorted correctly into their respective bins. We are then interested in verifying
two response time properties for the system: the first is $ \varphi_W $ which represents the system formula for sorting the white tokens and the second is $ \varphi_B $ which represents the system formula for sorting the blue tokens. We aim to monitor these formulae in a decentralized way. The notations used during the unwinding of both formulae are given in Table~\ref{table:io_var}. The two response time formulae of the system can be described as follows
$$\varphi_W = G ((LS_1 \land  SC) \circ_{\leq 5} A_{W}) 
$$
$$
 \varphi_B = G((LS_1 \land SC)\circ_{\leq 6} A_{B}) $$

\begin{table}[htbp]
	\centering
	\caption{I/O variables in the system}
	\begin{tabular}{||c|c||}
		\hline
		Variable & Definition  \\ [0.25ex]
		\hline\hline
		$ LS_{1/2} $      & Light Sensor 1, 2  \\
		$ T_{CS} $      & Trigger Color Sensor   \\
		$ CV_{W/B} $    & Annotated Color Value \\
		$ SC $          & Current Step Count        \\
		$ SC_{CP} $     & Token Step Count at CP     \\
		$ E_{W/B} $     & Bin ejection information to White/Blue bin      \\
		$ A_{W/B} $        & Arrival at White/Blue bin      \\ [0.25ex]
		
		\hline
	\end{tabular}
	\label{table:io_var}
\end{table}

We then unwind the two formulae $ \varphi_W $ and $ \varphi_B $ 
using the unwinding technique described at Section~\ref{Sec:unfoldingProcess}. 
In Table \ref{table:unwind_wht} we give the resulting sub-formulae
that result from the unwinding process of the formula $ \varphi_W $
and in Table \ref{table:unwind_blu} we give the resulting sub-formulae
that result from the unwinding process of the formula $ \varphi_B $
with their corresponding timing constraints.

\begin{table}
	\centering
	\caption{Formula Unwinding for the White Token}
	\begin{tabular}{||c|c|c|c||}
		\hline
		\# & Formula & c-value $ (seconds) $  & Process \\ [0.25ex]
		\hline\hline
		$ \varphi_W $     & $ LS_{1} \land SC \circ A_W $ & 4 & EC \\
		$ \phi_{W4} $    & $ E_W \circ A_W $ & 4 & EC \\
		$ \phi_{W3} $    & $ CV_{W} \land SC_{CP} \circ E_W $ & 2 & WBR \\
		$ \phi_{W2} $    & $ T_{CS}  \circ CV_{W} $ &  2 & WBR \\
		$ \phi_{W1} $    & $ LS_{1} \land SC \circ SC_{CP} $ & 2 & TD \\ 
		$ \phi_{*} $    & $ LS_{1} \land SC \circ T_{CS} $ & 1 & TD \\ [0.25ex]        
		\hline
	\end{tabular}
	\label{table:unwind_wht}
\end{table}

\begin{table}
	\centering
	\caption{Formula Unwinding for the Blue Token}
	\begin{tabular}{||c|c|c|c||}
		\hline
		\# & Formula & c-value $ (seconds) $ & Process \\ [0.25ex]
		\hline\hline
		$ \varphi_{B} $     & $ LS_{1} \land SC \circ A_B $ & 5 & EC \\
		$ \phi_{B4} $    & $ E_B \circ A_B $ & 5 & EC \\
		$ \phi_{B3} $    & $ CV_{B} \land SC_{CP} \circ E_B $ &  2 & BBR\\
		$ \phi_{B2} $    & $ T_{CS}  \circ CV_{B} $ &  2 & BBR \\
		$ \phi_{B1} $    & $ LS_{1} \land SC \circ SC_{CP} $ & 2 & TD \\ 
		$ \phi_{*} $    & $ LS_{1} \land SC \circ T_{CS} $ & 1 & TD \\ [0.25ex]        
		\hline
	\end{tabular}
	\label{table:unwind_blu}
\end{table}

\begin{table*} 
	\centering
	\caption{Faults and recovery plans}
	\begin{tabular}{||c|c||c||}
		\hline
		Formula & Fault & Recovery Plans \\ [0.25ex]
		\hline\hline
		$ \phi_* $          & Failure to trigger the color sensor $ (T_{CS}) $
		& Unclassified token goes into bin 3  \\ 
		$ \phi_{W1,B1} $    & Absence of the step count of the token $ (SC_{CP}) $     
		& Reference step count at $ LS_2 $ \\
		$ \phi_{W2,B2} $    & Color sensor output $ (CV_{W/B}) $ delay
		& Reduce conveyor belt speed \\
		$ \phi_{W3,B3} $    & WBR or BBR output $ (E_{W/B}) $ delay
		& Reduce conveyor belt speed \\
		$ \phi_{W4,B4} $    & Token fails to reach assigned bin $ (A_{W/B}) $
		& Token goes into bin 3   \\  [0.25ex]
		\hline
	\end{tabular} 
	\label{table:faults}
\end{table*}

Decentralized monitoring for each of the processes can now be done based on their sub-formulae. Note that in both Tables~\ref{table:unwind_wht} and \ref{table:unwind_blu}, there is a special sub-formula $ \phi_* $. $ \phi_* $ exists as a special observer as there is a separate requirement to trigger ($ T_{CS} $) the color processors at the moment the token is beneath the color sensor. Monitoring for each sub-formula are done at their respective processes, with the exception of sub-formulae $ \phi_{W2} $ and $ \phi_{B2} $. Instead, the observers of these sub-formulae are on WBR and BBR respectively, as WCP and BCP cannot host the 4DIAC runtime environment.

\subsection{The Failure Scenarios}

In this section, we describe some possible failure scenarios that may occur when running
the presented model factory. Depending on how early the failure is detected (i.e., violation of the main formula), different recovery plans may be performed to meet the main objective of the token sorting process. 

To enable failure recovery plans, a `clever' resilience mechanism is additionally employed to take advantage of the violations reported by the observers through monitoring of the sub-formulae in each process. A summary of the violations which can be detected through decentralized monitoring of the system is listed in Table~\ref{table:faults}. The faults listed in the table are based on timing deviations of their expected behavior.

\begin{enumerate} 
	\item 	
	{\textbf{Process / Formula: } Token detector/$ \phi_* $}.
	\textbf{Fault: } A failure to trigger the color sensor leads to an unclassified token even when the token is detected. Since color sorting is the main objective, a failure here constitutes to major fault of the system.
	\textbf{Recovery: } As a last resort, we push the unclassified token into bin 3 for re-sorting.
	
	\item 	
	\textbf{Process / Formulae: } Token detector/$ \phi_{W1/B1} $.
	\textbf{Fault: } It is plausible that the step count $ SC_{CP} $ read by the token detector is lost. Losing the step count undermines the system's ability to track the position of the token.
	\textbf{Recovery: } As this can be detected early by TD, we can respond by making use of the second light sensor $ LS_2 $ to redetermine the token's position and push the token into its rightful bin.
	
	\item 	
	\textbf{Process / Formulae: } Color processors / $ \phi_{W2/B2} $.
	\textbf{Fault: } As mentioned previously, the observers for the color processors (i.e., WCP/BCP) are hosted on WBR and BBR, respectively. A delay in the processing the color value of the token extends the execution latency of the color processors. 
	\textbf{Recovery: } To prevent the token from missing the ejector before a color value is produced, we can slow down the motor of the conveyor belt, allowing for more time for the color processors to produce an output.
	
	\item 	
	\textbf{Process / Formulae: } Bin resolvers / $ \phi_{W3/B3} $.
	\textbf{Fault: } The bin resolvers are responsible for assigning the token to their respective bins based on their color read by the color sensor. A delay in the making a decision for assigning the token extends the execution latency of the bin resolvers. 
	\textbf{Recovery: } To prevent the token from missing the ejector before a decision is reached, we can slow down the motor of the conveyor belt, allowing for more time for the bin resolvers to produce an output.
	
	\item 	
\textbf{Process / Formulae: } Ejector controller / $ \phi_{W4/B4} $.
	\textbf{Fault: } The sorted token has to reach its assigned bin. A fault occurs when the sorted token is unable to reach its assigned bin.
	\textbf{Recovery: } The token is pushed to bin 3 for re-sorting.
	
\end{enumerate}

The above recovery plans have been implemented on the above
described case study which allows the processes to respond
appropriately when a failure occurs (or an early violation of the properties is detected).
The goal of these recovery plans is to minimize the impact
of violations on the system by allowing the process to take the most appropriate possible action given the time at which (potential) violation of the main formula is detected. Without early violation detection, it is not possible to recover sufficiently in time to place the tokens into their respective bins. This is evident as seen in sub-formulae $ \phi_{W4/B4} $, which is equivalent to just monitoring the overall system formulae $ \varphi_{W/B} $. The tokens can only be ejected into bin 3 at the last point of violation detection.
A short  video documentation of the case study is available at \url{https://youtu.be/5CUH0Z2qaBM}.

 \section{Related Work}

 The key novelty of our presented  framework  comparing to existing frameworks \cite{Sen2004,BauerF12,ColomboF14,FalconeCF14,Scheffel14,MostafaB15,ColomboF16,Bataineh2019} is the employment of the tableau construction and the formula unwinding technique to split and distribute LTL formulas with quantitative operators so that monitoring of such class of properties can be conducted in a decentralised manner. The employment of these techniques allows  processes to detect early violations of properties and  perform some corrective or recovery actions to avoid severe consequences of failures.

Sen et al. \cite{Sen2004} propose a monitoring framework for safety properties of systems using the past-time linear temporal logic (ptLTL). However, the algorithm is unsound. The evaluation of some properties may be overlooked in their framework. This is because monitors gain knowledge about the state of the system by piggybacking  on the existing communication among processes. That is, if processes rarely communicate, then monitors exchange little information, and hence, some violations  may remain undetected. 
The authors   have not considered quantitative properties of  systems as we have done in this work and hence their framework cannot be directly applied to deal with this class of properties. 

Bauer and Falcone \cite{BauerF12} propose a decentralized framework for runtime monitoring of LTL. The framework is constructed from local monitors which can only observe the truth value of a predefined subset of propositional variables. The local monitors can communicate their observations in the form of a (rewritten) LTL formula towards its neighbors. 
The approach has the risk of 
saturating the communication devices as processes 
send their obligations as rewritten temporal formulas.
Mostafa and Bonakdarpour \cite{MostafaB15} propose similar decentralized LTL monitoring framework, but truth value of atomic variables
rather than rewritten formulas are shared. 
 Our work differs from these works in that we consider decentralised monitoring of
 quantitative properties where we extend the classical LTL with a quantitative dependency operator of the form $\circ_{\leq q}$. 
 The technical challenge of
monitoring quantitative properties in this setting consists of
translating global constraints into local ones which can be monitored by individual processes.

The work of Falcone et al. \cite{FalconeCF14} proposes a general decentralized monitoring algorithm in which the input specification is given as a deterministic finite-state automaton rather than an LTL formula. Their algorithm takes advantage of the semantics of finite-word automata, and hence they avoid the monitorability issues induced by the infinite-words semantics of LTL. They show that their implementation outperforms the Bauer and Falcone decentralized LTL algorithm \cite{BauerF12} using several monitoring metrics. It is not clear to us how the decentralised monitoring framework of  \cite{FalconeCF14}  based on finite state automata can be used to monitor quantitative LTL properties in which costs may accumulate from one state to another. 
Our framework employs also an unwinding algorithm which helps to optimise decentralised monitoring of  properties  in  a  way  such  that  violations  can be detected way before the original property would fail. Early detection of violations cannot be achieved using the monitoring framework of \cite{FalconeCF14}.

 Colombo and Falcone \cite{ColomboF16} propose a new way of organizing monitors called choreography, where monitors are organized as a tree across the distributed system, and each child feeds intermediate results to its parent. The proposed approach tries to minimize the communication induced by the distributed nature of the system and focuses on how to automatically split an LTL formula according to the architecture of the system.
 However, their framework cannot be used to efficiently monitor  LTL properties with quantitative operators. 
The key difference between our framework and their framework is the employment of the tableau construction and formula unwinding technique to split and distribute the global quantitative constraint, which help to detect early violations of the monitored property and  perform some recovery actions.

  Recently, Al-Bataineh et al. \cite{Bataineh2019} 
presented a  monitoring framework
for LTL formulas in which functional properties are modeled as LTL formulas and decomposed using the tableau decomposition rules.
They showed how to use tableau to optimise the underlying decentralized monitoring process of functional properties for synchronous distributed  systems.  However, quantitative LTL properties
 have not been considered in their work and hence cumulative cost properties cannot be efficiently monitored in their decentralised framework.

Several extensions to the classical LTL (both past and future LTL)  have been proposed in the prior literature \cite{Hubert2000,DEMRI2007380,Bersani2010}.
The authors of \cite{Hubert2000,DEMRI2007380,Bersani2010}  focused mainly on showing decidability of some restricted forms of constraint systems using automata-theoretic technique. They extended  until with arithmetic expressions with integer variables, which maybe used to model quantitative properties of systems. 
However, in our work we extend the LTL with a quantitative dependency operator of the form $\circ_{\leq q}$ which can be used to capture quantitative dependencies among variables/modules in the system model.
Such extension allows us to monitor in a straightforward manner an interesting class of quantitative properties, namely cumulative cost properties of systems.
The introduced quantitative dependency operator allows direct verification of cumulated costs among dependent parts (modules, processes, or variables) of the system being monitored, where the left and right operands of the dependency operator can be atomic formula, compound formula, or LTL formula.

\section{Conclusion}

The topic and idea of splitting monitoring of systems into simpler monitoring tasks is an interesting current research problem, especially when considering upcoming applications like cloud, edge and fog computing.
In this work, we introduced a methodology to decentralize the monitoring of cumulative cost properties of systems 
formalised as temporal properties and represented as a tableau. 
 The decentralization process works by systematically transforming (``unwinding'')
the system level LTL formula into a semantically equivalent formula which can then be decomposed and
distributed across the system's processes/nodes in order to monitor
fulfilment of the respective sub-properties at runtime. If a monitor detects
a violation, depending on its nature, the error can either be forwarded to
a superordinate process or corrective
actions may be initiated directly at the level of the process detecting
the fault. As such violations can typically be detected way before the
original property would fail, such corrective actions can even
avoid the system failure in some cases.
The methodology is demonstrated
with two synthetic examples and a real experiment involving a Fischertechnik
plant model.

\section*{Acknowledgement}
This work was supported by Delta-NTU Corporate Lab for Cyber-Physical Systems with funding support from Delta Electronics Inc. and the National Research Foundation (NRF) Singapore under the Corp Lab@University Scheme.

 \bibliographystyle{IEEEtran}
 \bibliography{references}

\end{document}